\documentclass[version=final]{iacrcc}

\license{CC-by}

\usepackage{graphicx}
\usepackage{braket}
\usepackage{amsmath} %
\usepackage{amssymb} %
\usepackage{todonotes}
\usepackage{array}
\usepackage{algorithm}
\usepackage{algorithmic}
\usepackage{xspace}
\usepackage{framed}
\usepackage{ulem}
\usepackage{extarrows}
\usepackage{latexsym}
\usepackage{bm}
\usepackage[all]{xy}
\usepackage{amsfonts}
\usepackage{mathrsfs}
\usepackage{booktabs}
\usepackage{lipsum}                     
\usepackage{xargs}

\newcommand{\Ext}{\ensuremath{\mathsf{Ext}}}

\newcommand{\negl}{\ensuremath{\mathsf{negl}}\xspace}

\newcommand{\crs}{\ensuremath{\mathsf{crs}}}

\allowdisplaybreaks[1]

\newcommand{\NP}{\mathsf{NP}}

  \newcommand{\cA}{\ensuremath{{\mathcal A}}\xspace}

  \newcommand{\cR}{\ensuremath{{\mathcal R}}\xspace}

\newcommand{\zo}{\ensuremath{{\{0,1\}}}\xspace}

\newcommand{\cL}{\ensuremath{{\mathcal L}}\xspace}

\newcommand{\sP}{\mathsf{P}}
\newcommand{\sV}{\mathsf{V}}

\newcommand{\Sim}{\mathsf{Sim}}

\newcommand{\td}{\mathsf{td}}

\newcommand{\bbN}{\mathbb{N}}

\newcommand{\commentout}[1]{}

\newcommand{\Setup}{\ensuremath{\mathsf{Setup}}\xspace}

\newcommand{\token}{\texttt{tok}}

\usepackage[nointegrals]{wasysym}

\newcommand{\R}{\ensuremath{\mathcal{R}}}

\newcommand{\X}{\ensuremath{\mathcal{X}}}
\newcommand{\Y}{\ensuremath{\mathcal{Y}}}

\newcommand{\supp}{\textsc{Supp}}
\newcommand{\genf}{\textsc{Gen}_{\mathcal{F}}}
\newcommand{\invf}{\textsc{Inv}_{\mathcal{F}}}
\newcommand{\chkf}{\textsc{Chk}_{\mathcal{F}}}
\newcommand{\sampf}{\textsc{Samp}_{\mathcal{F}}}

\def\*#1{\mathbf{#1}}

\newtheorem{claim}{Claim}

\renewenvironment{claim}
  {\par\noindent\refstepcounter{claim}\textbf{Claim \theclaim.}\,\itshape}
  {\par}
\title
 [running  = {Zero-Knowledge Proofs of Quantumness},
      ]
{Zero-Knowledge Proofs of Quantumness}
\addauthor[
           orcid   = {0000-0003-1136-4064},
           inst    = {1},
           email   = {hieu.phan@telecom-paris.fr}
          ]{Duong Hieu Phan}

\addauthor[
           orcid   = {0000-0001-5272-2572},
           inst    = {1},
           email   = {weiqiang.wen@telecom-paris.fr}
          ]{Weiqiang Wen}

\addauthor[
           orcid   = {0000-0001-9147-6554},
           inst    = {2,3},
           email   = {yanxy2020@bupt.edu.cn}
          ]{Xingyu Yan}

\addauthor[
           inst    = {1},
           email   = {jinwei.zheng@telecom-paris.fr}
          ]{Jinwei Zheng}

\addaffiliation[country={France}]{LTCI, Telecom Paris, Institut Polytechnique de Paris, Paris}
\addaffiliation[country={China}]{State Key Laboratory of Networking and Switching Technology, Beijing University of Posts and Telecommunications, Beijing, 100876}
\addaffiliation[country={China}]{School of Cyberspace Science and Technology, Beijing Institute of Technology, Beijing, 100081}

\genericfootnote{This work was supported in part by the National Key Research and Development Program of China (No. 2022YFB2702700), by the European Union's Horizon Europe research and innovation programme under the project ``Quantum Secure Networks Partnership'' (QSNP, No. 101114043), by the France 2030 ANR Project ANR-22-PECY-003 SecureCompute, and by the ANR Project ANR-24-CE48-4293 HELO.}
\begin{document}

\maketitle


\begin{abstract}
  With the rapid development of quantum computers, proofs of quantumness have recently become an interesting and intriguing research direction. However, in all current schemes for proofs of quantumness, quantum provers almost invariably face the risk of being maliciously exploited by classical verifiers. In fact, through malicious strategies in interaction with quantum provers, classical verifiers could solve some instances of hard problems that arise from the specific scheme in use. In other words, malicious verifiers can break some schemes (that quantum provers are not aware of) through interaction with quantum provers. All this is due to the lack of formalization that prevents malicious verifiers from extracting useful information in proofs of quantumness.

To address this issue, we formalize zero-knowledge proofs of quantumness. Intuitively, the zero-knowledge property necessitates that the information gained by the classical verifier from interactions with the quantum prover should not surpass what can be simulated using a simulated classical prover interacting with the same verifier. As a result, the new zero-knowledge notion can prevent any malicious verifier from exploiting quantum advantage.
Interestingly, we find that the classical zero-knowledge proof is sufficient to compile some existing proofs of quantumness schemes into zero-knowledge proofs of quantumness schemes.

Due to some technical reason, it appears to be more general to require zero-knowledge proof on the verifier side instead of the prover side. Intuitively, this helps to regulate the verifier's behavior from malicious to be honest-but-curious. As a result, both parties will play not only one role in the proofs of quantumness but also the dual role in the classical zero-knowledge proof. 

Specifically, the two principle proofs of quantumness schemes: Shor's factoring-based scheme and learning with errors-based scheme in [Brakerski et al, FOCS, 2018], can be transformed into zero-knowledge proofs of quantumness by requiring an extractable non-interactive zero-knowledge argument on the verifier side. 
Notably, the zero-knowledge proofs of quantumness can be viewed as an enhanced security notion for proofs of quantumness. To prevent malicious verifiers from exploiting the quantum device's capabilities or knowledge, it is advisable to transition existing proofs of quantumness schemes to this framework whenever feasible.
\end{abstract}

\begin{textabstract}
  With the rapid development of quantum computers, proofs of quantumness have recently become an interesting and intriguing research direction. However, in all current schemes for proofs of quantumness, quantum provers almost invariably face the risk of being maliciously exploited by classical verifiers. In fact, through malicious strategies in interaction with quantum provers, classical verifiers could solve some instances of hard problems that arise from the specific scheme in use. In other words, malicious verifiers can break some schemes (that quantum provers are not aware of) through interaction with quantum provers. All this is due to the lack of formalization that prevents malicious verifiers from extracting useful information in proofs of quantumness.

To address this issue, we formalize zero-knowledge proofs of quantumness. Intuitively, the zero-knowledge property necessitates that the information gained by the classical verifier from interactions with the quantum prover should not surpass what can be simulated using a simulated classical prover interacting with the same verifier. As a result, the new zero-knowledge notion can prevent any malicious verifier from exploiting quantum advantage.
Interestingly, we find that the classical zero-knowledge proof is sufficient to compile some existing proofs of quantumness schemes into zero-knowledge proofs of quantumness schemes.

Due to some technical reason, it appears to be more general to require zero-knowledge proof on the verifier side instead of the prover side. Intuitively, this helps to regulate the verifier's behavior from malicious to be honest-but-curious. As a result, both parties will play not only one role in the proofs of quantumness but also the dual role in the classical zero-knowledge proof. 

Specifically, the two principle proofs of quantumness schemes: Shor's factoring-based scheme and learning with errors-based scheme in [Brakerski et al, FOCS, 2018], can be transformed into zero-knowledge proofs of quantumness by requiring an extractable non-interactive zero-knowledge argument on the verifier side. 
Notably, the zero-knowledge proofs of quantumness can be viewed as an enhanced security notion for proofs of quantumness. To prevent malicious verifiers from exploiting the quantum device's capabilities or knowledge, it is advisable to transition existing proofs of quantumness schemes to this framework whenever feasible.
\end{textabstract}

\section{Introduction}
A cryptographic proofs of quantumness (PoQ) is a scheme that validates a quantum device's ability to perform tasks that classical devices cannot handle efficiently. These schemes enable a quantum prover to convince a classical verifier of its quantum computational power through specific challenging tasks, such as factoring~\cite{shor1994algorithms,regev2023efficient}, BosonSampling~\cite{aaronson2011computational,bremner2011classical}, learning with errors (LWE)-based trapdoor claw-free (TCF) hash functions~\cite{brakerski2021cryptographic,brakerski2020simpler}, computational Bell test~\cite{KCVY,KLVY,brakerski2023simple}, etc. 

Recent PoQ schemes primarily address scenarios where the prover may act maliciously, such as checking whether a classical prover is masquerading as a quantum device and verifying whether the results produced by the quantum device are indeed correct. In addition to this, this work considers an important supplementary perspective, focusing on situations where \textbf{the classical verifier might act maliciously}. To see how a malicious verifier could behave, we provide a scenario about how the quantum prover could be maliciously exploited in the PoQ as shown in Figure~\ref{fig.0}.

In the factoring-based PoQ scheme, a malicious verifier~$\mathcal{V}^*$ can simply replace the large integer number~\( N \) with their carefully chosen number~\( N^* \) (e.g., that can be taken from an RSA public key). With such substitution, 
the verifier could easily exploit the prover's quantum capabilities to factor~\(N^*\) for gaining benefit.
Then in the LWE-based PoQ scheme in~\cite{brakerski2021cryptographic}, how the verifier can make use of the quantum prover is much less clear. On one side, assuming the post-quantum hardness of LWE problem, the classical verifier should not be able to get much more advantage from the quantum prover for solving the LWE instance. Whereas the transcripts from the quantum prover should contain more information than that which any classical prover could possess. This actually helps the verifier to distinguish quantum prover from any macilious classical prover. To conclude, we can not properly prevent the verifier from obtaining additional information through the quantum prover beyond being merely convinced of the prover's quantumness.

Regarding the explicit attack and potential unexpected advantage that the verifier could gain from the quantum prover, preventing malicious verifiers from extracting useful information from PoQ is necessary and raises the following question:
\begin{center}
   \textit{Is there a ``zero-knowledge'' proofs of quantumness that allows a prover to demonstrate quantum capability without revealing any other useful information?}
\end{center}

The ``zero-knowledge'' property is to ensure that, at any point, the quantum prover's capabilities are not maliciously exploited by the verifier. 
In other words, the interactive proof process does not reveal any information beyond the fact that ``the prover possesses quantum capabilities''. 
We believe that considering the zero-knowledge PoQ (ZKPoQ) scheme is crucial and has far-reaching implications. In the post-quantum era, the test of quantumness may serve as the first step for classical users to subscribe to quantum services. The ZKPoQ can fundamentally ensure that the quantum server's computational power is not swindled by classical users during this verification process, thereby safeguarding the server’s interests.

\begin{figure}[t]
\centering
\includegraphics[width=0.98\textwidth]{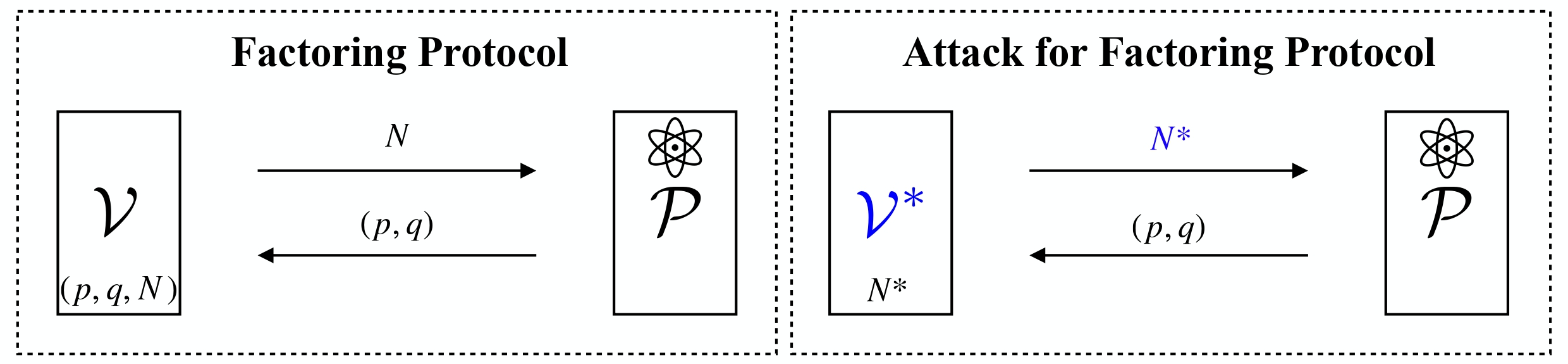}
\caption{
Schematic diagram of a malicious verifier~$\mathcal{V}^*$ spoofing to utilize the quantum capabilities of the quantum prover~$\mathcal{P}$ in factoring-based quantumness proof scheme.
} \label{fig.0}
\end{figure}


\subsection{Our Contributions}\label{sec:contribution}

The traditional security model for the proofs of quantumness only concerns the completeness and soundness, where the latter one requires that no classical prover can successfully run the protocol and pass the verification.
In this work, we introduce the zero-knowledge property for the proofs of quantumness protocol for the first time. In particular, the zero-knowledge property is also defined in a classical manner but with respect to quantum capability instead of particular knowledge. Informally, a proofs of quantumness protocol is said to be zero-knowledge when the communication between the quantum prover and the classical verifier can be perfectly simulated with a polynomial-time probabilistic classical prover communicated with the same verifier (as shown in Figure~\ref{fig:ZK-idea}). Under this zero-knowledge definition, the verifier can not exploit the quantum advantage of the quantum prover via their communication anymore, as their communication does not leak more information than the communication between the same verifier and a (simulated) classical prover.

\begin{figure}[t]
\centering
\includegraphics[width=0.98\textwidth]{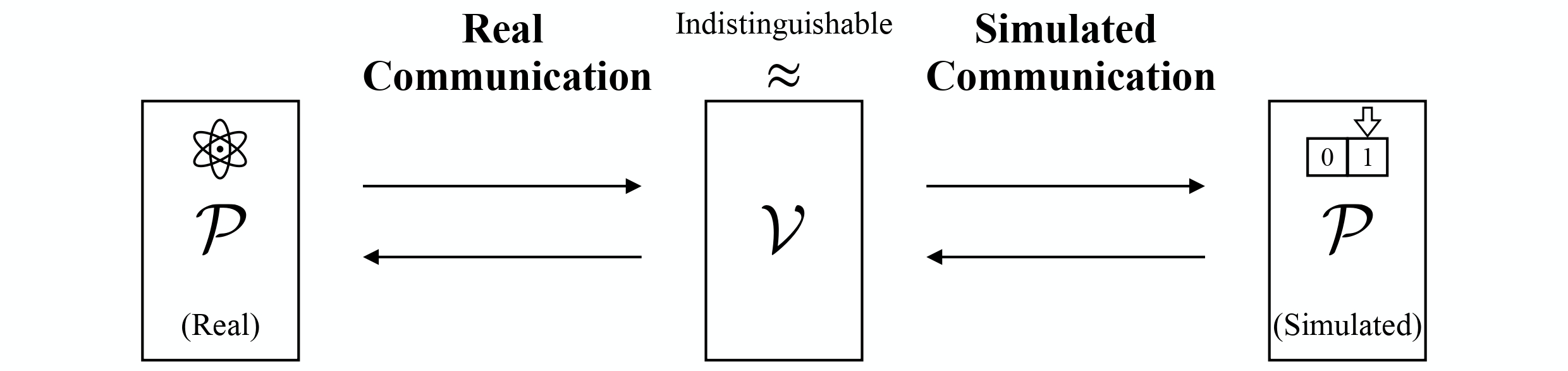}
\caption{
Illustration of our new zero-knowledge property for proofs of quantumness.
} \label{fig:ZK-idea}
\end{figure}

To demonstrate the rationality of this new notion, we migrate two mainstream PoQ schemes, factoring-based scheme and LWE-based scheme~\cite{brakerski2021cryptographic}, to the ones satisfying zero-knowledge. 
Informally, we define the PoQ with zero knowledge property as follows.
\begin{definition}[Zero-knowledge proofs of quantumness, ZKPoQ]
 A ZKPoQ scheme satisfies quantum completeness, classical soundness, and (computational) zero-knowledge properties.
\end{definition}

Our first technical contribution (the formal statement is in Section~\ref{Sec:ZKPOQ demonst.}) concerns the zero-knowledge property. We show that certain existing PoQ schemes can be efficiently transformed into the ones with zero-knowledge property.

\paragraph{\textbf{Infeasibility of requiring zero-knowledge proofs from prover}.} Intuitively, one might think that a quantum prover could directly make use of a zero-knowledge proof to hide the information from the prover while still making the verification feasible. Indeed, this approach works well with the factoring-based scheme, where the prover could directly respond with zero-knowledge proof of factorization~$(p,q)$ for~$N$. However, it is not suitable for other schemes such as the LWE-based PoQ scheme~\cite{brakerski2021cryptographic}. In the latter one, 
the verifier will challenge the quantum prover to provide the witness to some statement that is only known by the verifier itself. In this case, it is infeasible to require the prover to provide a zero-knowledge proof without knowing the explicit statement to prove.
In more details, in the equation test of the PoQ scheme in~\cite{brakerski2021cryptographic} (refer to Figure~\ref{fig:BCMVV} for the description of the scheme), the prover is asked to provide a witness~\( \sigma_1=(c, \mathbf{d}) \in \{0,1\}\times\{0,1\}^n \) to the statement~\( c = \mathbf{d}^\top \cdot (\textbf{x}_0 \oplus \textbf{x}_1) \bmod 2 \), where~$\textbf{x}_0$ and~$\textbf{x}_1$ are two preimages for a same output~$y$ under a post-quantum secure claw-free function. Due to the claw-free property, even the quantum prover should not be able to obtain both preimages~\( \textbf{x}_0 \) and~\( \textbf{x}_1 \). As a result, the adversary does not have access to the statement.
Therefore, instead of requiring the prover to provide a zero-knowledge proof, our approach considers ensuring ZKPoQ by requiring the verifier to provide a zero-knowledge proof to claim that it is not behaving maliciously.

This leads to our first technical contribution, which supplements zero-knowledge on top of two existing PoQ schemes by requiring an extractable proof from the verifier's side. 
To sketch the idea, any classical prover communicated with the verfier, can now extract necessary secret knowledge from the additional extractable proof from the verifier, which is sufficient to be used to simulate the communication without any quantum resource.
This properly ensures the zero-knowledge on the prover side for the PoQ scheme. However, the newly added extractable proof from the verifier causes an issue on the soundness of the PoQ protocol, as any classical prover can simulate the communication using the knowledge under the extractable proof without exploiting quantum advantage. To fix this issue, we further require the extractable proof to be zero-knowledge. The latter one can be constructed from one-way function and public key encryption~\cite{JK23}. 
As an interesting result, both zero-knowledge and soundness properties of the aforementioned ZKPoQ primitive can be satisfied simultaneously. In particular, we provide transformation from two PoQ schemes to ZKPoQ schemes (refer to Theorems~\ref{thm:informal_factoring} and~\ref{thm:informal_LWE} for the informal statements of our results).

\begin{figure}[t]
\centering
\includegraphics[width=0.98\textwidth]{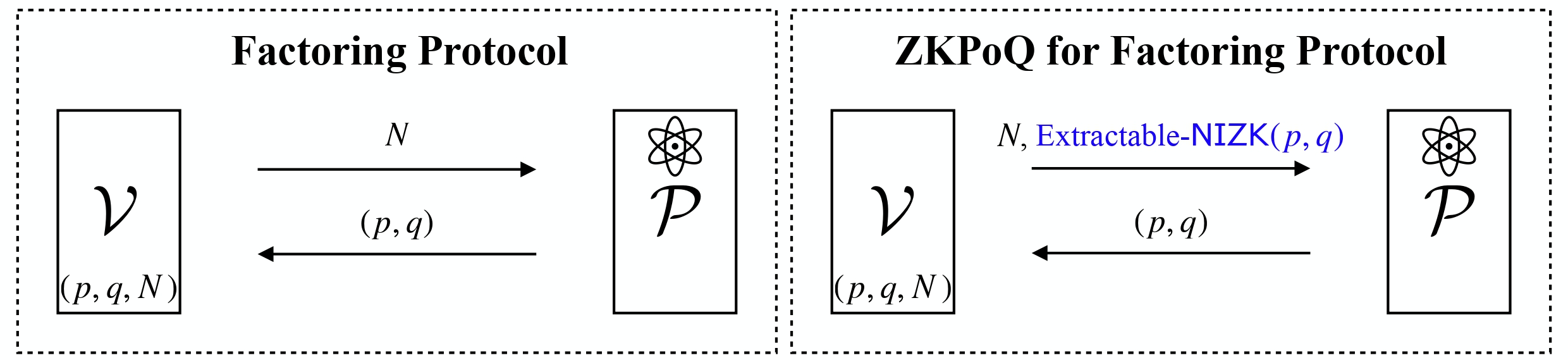}
\caption{Diagram of ZKPoQ for factoring-based scheme. } 
\label{fig.example1}
\end{figure}

As shown in Figure~\ref{fig.example1}, the main modification on top of the original scheme is that we further require the verifier to provide an extractable-NIZK proof of the factorization~$(p, q)$ of the integer~$N$. 
One might notice that both the prover and the verifier in the proofs of quantumness scheme will also play the dual role in the additional NIZK proof. In particular, to argue the soundness of this resulting ZKPoQ scheme, we can rely on the zero-knowledge property of the extractable-NIZK proof such that any malicious classical prover can not get any information from this additional proof and should factor the integer~$N$ by itself. 
Notably, a classically secure NIZK will be sufficient as the ZK property is only required to hold facing a malicious classical prover for the soundness of proofs of quantumness.
Regarding the zero-knowledge property of this ZKPoQ scheme, the simulator can then make use of the extractability property of the extractable-NIZK proof to extract the witness~$(p, q)$ from the proof. Therefore, the transcript from the quantum prover can be simulated perfectly.

\begin{theorem}[ZKPoQ for factoring-based scheme]\label{thm:informal_factoring}
Given an extractable-NIZK proof, the factoring-based PoQ scheme shown in Figure~\ref{fig.example1} can be transformed into ZKPoQ one.
\end{theorem}

Now we can see how this ZKPoQ scheme can be used to avoid the attack scenarios depicted in Figure~\ref{fig.0}.
As illustrated in Figure~\ref{fig.example1}, according to the extractability of the extractable-NIZK proof, it would be infeasible for a malicious verifier~\( \mathcal{V}^* \) to provide a valid proof without knowing the factorization~$(p, q)$, which efficiently regulate the malicious behavior of the verifier in this scenario.

Next, we move the transformation from the LWE-based scheme in Figure~\ref{fig.example2} to ZKPoQ scheme. Again, the principle adjustment is to ask the verifier to submit an extractable-NIZK proof of the secret of the LWE instance. As before, the soundness of the ZKPoQ scheme is based on the zero-knowledge property of the extractable-NIZK proof. Notably, the zero-knowledge property of our LWE-based ZKPoQ scheme can also be achieved by the extratability of the extractable-NIZK proof. The main observation is that any classical prover knowing the secret of the LWE instance (sent from the verifier) can also pass the verification. Therefore, the simulator can extract the LWE secret and then perfectly simulate the transcripts from the prover to the verifier. This does not invalidate the soundness of the scheme, because the verifier is the only one knowing the secret, and therefore anyone else should follow the scheme to prove quantumness.

\begin{figure}[htbp]
\centering
\includegraphics[width=0.98\textwidth]{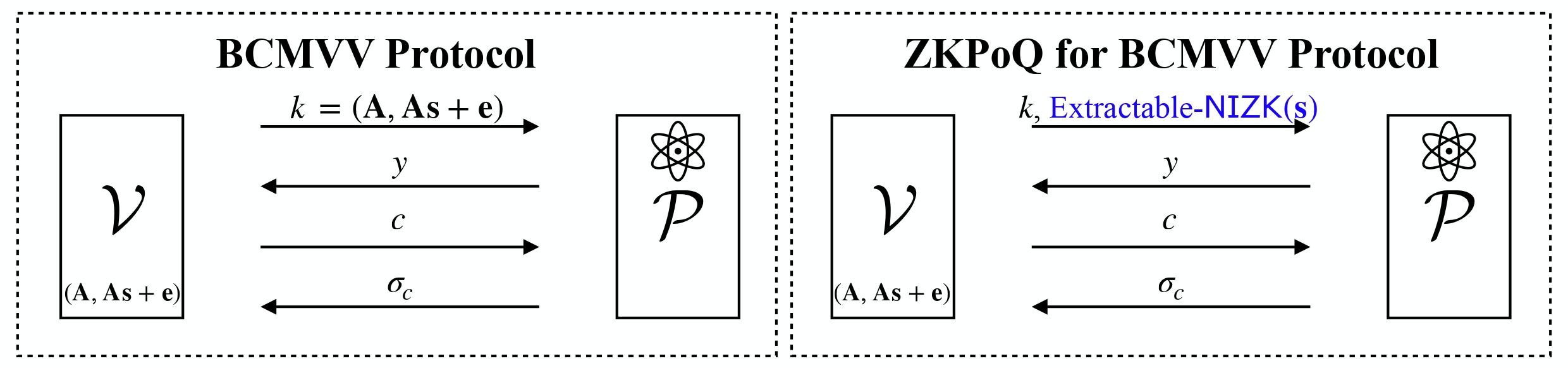}
\caption{Diagram of ZKPoQ for the LWE-based scheme~\cite{brakerski2021cryptographic}.} 
\label{fig.example2}
\end{figure}

\begin{theorem}[ZKPoQ for the LWE-based scheme~\cite{brakerski2021cryptographic}]~\label{thm:informal_LWE}
Given an extractable-NIZK proof, the LWE-based scheme in~\cite{brakerski2021cryptographic} can be transformed into ZKPoQ one.
\end{theorem}

Last but not least, we emphasize that using classically secure extractable-NIZK proof in the aforementioned ZKPoQ schemes is sufficient. Moreover, if we want to apply our LWE-based ZKPoQ scheme for other purpose such as certifiable randomness from quantum device~\cite{brakerski2021cryptographic} or key leasing~\cite{chardouvelis2023quantum}, post-quantum secure extractable-NIZK seems to be necessary. Because in that case, we want to ensure that even the quantum prover could not gain any extra advantage from the extractable-NIZK proof. As shown in Theorem~\ref{cor:simext-nizk-crs}, there exists a post-quantum secure extractable-NIZK proof based on the LWE assumption.

\paragraph{\textbf{Difficulty of having a generic transformation for proofs of quantumness to be zero-knowledge}.} 
In this work, we have examinated two examples based on factoring and LWE, respectively. 
These two protocols shares the same challenge-response type of procedure between the prover and verifier. Under this category, there exist more LWE-based protocols, but additional assumptions are involved. For example, the PoQ scheme in~\cite{brakerski2020simpler} requires random oracle and the ones in~\cite{KCVY,brakerski2023simple} further assume Bell's inequality. Therefore, it is not clear whether the idea of requiring an extractable-NIZK from the verifier's side is also applicable to the other LWE-based protocols. Other than the challenge-response style protocols, there is also non-interactive PoQ protocol such as the Bosonsampling~\cite{aaronson2011computational,bremner2011classical}. Therefore, though our idea seems to be quite general, different methods may still be required for other protocols to achieve zero-knowledge.

\subsection{Related works}\label{sec:relatedworks}

Proofs of quantumness (PoQ) represent a milestone for the field of quantum computation, serving as a classically-verifiable demonstration of non-classical behavior from a single quantum device. To date, there are three mainstream approaches for demonstrating proofs of quantumness:
\begin{enumerate}
    \item \textbf{Factoring-Based PoQ}: This approach leverages the well-known computational hardness of factoring large integers, a task presumed to be intractable for classical computers but efficiently solvable by quantum computers using algorithms such as Shor’s algorithm~\cite{shor1994algorithms}. The quantum demonstration involves producing factors of large composite numbers, which can be verified classically. However, implementing factoring-based PoQ requires fault-tolerant scalable quantum computer, which is unimplementable on short-term NISQ devices.

    \item \textbf{Sampling-Based PoQ}: This approach is based on the classical hardness of sampling problems, such as boson sampling~\cite{aaronson2011computational,madsen2022quantum} or random circuit sampling~\cite{Arute2019quantum, wu2021strong}. Although these methods can be implemented with NISQ hardware, they are not efficiently verifiable since the classical verification of these methods typically demands exponential time.

    \item \textbf{LWE-Based PoQ}: Another PoQ approach introduced by Brkerski et al.~\cite{brakerski2021cryptographic}, utilizes the quantum hardness of the Learning With Errors (LWE) problem to execute a cryptographic interactive protocol between a  polynomial-time  classical verifier and an ostensibly quantum polynomial-time prover. Wherein, the verifier presents challenges to the prover and checks the correctness of the prover’s responses. A crucial aspect of this method is that an effective quantum strategy should enable the prover to accurately respond to the verifier’s challenges with a high likelihood of success, while any effective classical strategy would only achieve a low probability of success, assuming the hardness of LWE assumption. In terms of implementation, the LWE-based PoQ appears to require significantly fewer resources compared to the Factoring-based PoQ, yet it still demands more than the Sampling-based PoQ. Regarding verification, both LWE-based and Factoring-based PoQ can be efficiently verified classically, unlike Sampling-based PoQ, which does not allow for such efficient classical verification.

    Following on the breakthrough work of~\cite{brakerski2021cryptographic}, numerous studies such as~\cite{brakerski2020simpler, alagic2020non, KCVY, KLVY, alnawakhtha2024lattice,brakerski2023simple} have developed progressively more efficient proofs of quantumness. The goal of these efforts is to simplify these tests so that they can be implemented on current NISQ quantum devices.

\end{enumerate}
 The above PoQ protocols, along with follow-up works, have not addressed scenarios that involve malicious verifiers, an area yet to be explored in the current research. This paper introduces zero-knowledge proofs of quantumness, presenting them as an advanced security concept for proofs of quantumness. Given that sampling-based PoQ schemes do not align with the cryptographic interactive protocol framework and lack the capability for efficient classical verification, our work specifically focuses the formalization of zero-knowledge proofs on classically verifiable PoQ protocols, with particular emphasis on those based on factoring and LWE.

\subsection{Open Problems}
Although our results naturally transform PoQ into ZKPoQ using extractable-NIZK, there remain many intriguing open problems regarding the general transformation of PoQ into ZKPoQ. Here, we mention some of them.
\begin{itemize}
    \item Is it possible to provide a more general transformation framework that formalizes the required characteristics of the PoQ? In our current work, we have demonstrated that both the factoring-based PoQ and LWE-based PoQ protocols can be directly upgraded to ZKPoQ using extractable-NIZK techniques. However, for sampling-based PoQ schemes, they do not align with the cryptographic interactive protocol framework and lack the capability for efficient classical verification, which makes such a general transformation challenging. This remains an open problem, and we will consider exploring this direction in future work. Additionally, for other LWE-based PoQ protocols such as those presented in~\cite{brakerski2020simpler, alagic2020non, KCVY, KLVY, alnawakhtha2024lattice,brakerski2023simple}, we note that these protocols often rely on different heuristic assumptions.  Therefore, it is not trivial to apply our approach to these schemes. However, these protocols closely resemble the one we considered, making it interesting to explore potential common characteristics that satisfy the transformation. This would be another future research direction. 
    \item Can interactive solutions in the plain model replace NIZKs in both ZKPoQ constructions, thereby eliminating the need for the CRS model? In our current construction, we chose to use post-quantum NIZK to minimize interaction between parties, but this comes at the cost of relying on a CRS. In order to remove the CRS and work in the plain model, one can instead resort to the interactive zero-knowledge proofs~\cite{goldreich1986proofs,bellare1992defining}. We will consider it as future work to formally explore the feasibility of this replacement. 
\end{itemize}

\section{Preliminaries}
\subsection{Notation}
We use the acronyms PPT and QPT for probabilistic polynomial time and quantum polynomial time respectively.
For a classical probabilistic algorithm~$\mathcal{A}$, we write~$\mathcal{A}(x ; r)$ to denote running~$\mathcal{A}$ on input~$x$, with input randomness~$r$. For a finite set~$S$, we use~$x \xleftarrow{\$} S$ to denote uniform sampling of~$x$ from the set~$S$. 
We denote~$[n] = \{1, 2, \cdots, n\}$. 
For  clarity, unless otherwise stated, the terms ``Verifier ($\mathcal{V}$)'' and ``Prover ($\mathcal{P}$)'' refer specifically to their roles in the PoQ scheme.

Conceptually, trapdoor claw-free functions (TCF) consist of a pair of injective functions~\(\{f_{k,0},f_{k,1}\}_k\) that share the same image. With access to a secret trapdoor~\(\td\), it becomes easy to determine the two preimages~\(\*x_0\) and~\(\*x_1\) of the same image~\(y\), such that~\(f_0(\*x_0) = f_1(\*x_1) = y\). However, it is computationally difficult to invert~\(\{f_{k,0},f_{k,1}\}_k\) without the trapdoor~$\td$. Such a pair of~$(\*x_0,\*x_1)$ is known as a claw, hence the name is claw-free. In the LWE-based scheme~\cite{brakerski2021cryptographic}, they use the Noisy TCF (NTCF) informally described by~$f'_{k,b}(\*x)=\*{Ax}+\*e'+b\cdot (\*{As}+\*e)$ for~$b \in \{0,1\}$ and~$k=(\*A, \*A\*s+\*e)$. In the following, we recall the definition of NTCF as well as its realization under LWE.
\begin{definition}[NTCF family,~\cite{brakerski2021cryptographic}]\label{def: ntcfs}
	Let $\lambda$ be a security parameter. Let $\X$ and $\Y$ be finite sets.
	 Let $\mathcal{K}_{\mathcal{F}}$ be a finite set of keys. A family of functions 
	$$\mathcal{F} \,=\, \big\{f_{k,b} : \X\rightarrow \mathcal{D}_{\Y} \big\}_{k\in \mathcal{K}_{\mathcal{F}},b\in\{0,1\}}$$
	is called a noisy trapdoor claw-free (NTCF) family if the following conditions hold:

	\begin{enumerate}
		\item{\textbf{Efficient Function Generation.}} There exists an efficient probabilistic algorithm $\genf$ which generates a key $k\in \mathcal{K}_{\mathcal{F}}$ together with a trapdoor $t_k$: 
		$$(k,t_k) \leftarrow \genf(1^\lambda)\;.$$

		\item{\textbf{Trapdoor Injective Pair.}}  
		\begin{enumerate}
			\item \textit{Trapdoor}: There exists an efficient deterministic algorithm $\invf$ such that with overwhelming probability over the choice of $(k,t_k) \gets \genf(1^{\lambda})$, the following holds: \\
			$$ \text{for all } b\in \{0,1\}, x\in \X \text{ and } y\in \supp(f_{k,b}(x)), \invf(t_k,b,y) = x. $$
			\item \textit{Injective pair}: For all keys $k\in \mathcal{K}_{\mathcal{F}}$, there exists a perfect matching $\R_k \subseteq \X \times \X$ such that $f_{k,0}(x_0) = f_{k,1}(x_1)$ if and only if $(x_0,x_1)\in \R_k$. 
		\end{enumerate}

		\item{\textbf{Efficient Range Superposition.}}
		For all keys $k\in \mathcal{K}_{\mathcal{F}}$ and $b\in \{0,1\}$ there exists a function $f'_{k,b}:\X\to \mathcal{D}_{\Y}$ such that the following hold.
		\begin{enumerate} 
			\item For all $(x_0,x_1)\in \mathcal{R}_k$ and $y\in \supp(f'_{k,b}(x_b))$, $\invf(t_k,b,y) = x_b$ and $\invf(t_k,b\oplus 1,y) = x_{b\oplus 1}$. 
			\item There exists an efficient deterministic procedure $\chkf$ that, on input $k$, $b\in \{0,1\}$, $x\in \X$ and $y\in \Y$, returns $1$ if $y\in \supp(f'_{k,b}(x))$ and $0$ otherwise. Note that $\chkf$ is not provided the trapdoor $t_k$. 
			\item For every $k$ and $b\in\{0,1\}$, 
			$$ \mathbb{E}_{x\leftarrow \X} \big[\,H^2(f_{k,b}(x),\,f'_{k,b}(x))\,\big] \,\leq\, \mu(\lambda) \;.$$
			 for some negligible function $\mu$. Here $H^2$ is the Hellinger distance. Moreover, there exists an efficient procedure  $\sampf$ that on input $k$ and $b\in\{0,1\}$ prepares the state
			\begin{equation*}
			    \frac{1}{\sqrt{|\X|}}\sum_{x\in \X,y\in \Y}\sqrt{(f'_{k,b}(x))(y)}\ket{x}\ket{y}\;.
			\end{equation*}
		\end{enumerate}

\item{\textbf{Adaptive Hardcore Bit.}}
For all keys $k\in \mathcal{K}_{\mathcal{F}}$ the following conditions hold, for some integer $w$ that is a polynomially bounded function of $\lambda$. 
\begin{enumerate}
\item For all $b\in \{0,1\}$ and $x\in \X$, there exists a set $G_{k,b,x}\subseteq \{0,1\}^{w}$ such that $\Pr_{d\leftarrow_U \{0,1\}^w}[d \notin G_{k,b,x}]$ is negligible, and moreover there exists an efficient algorithm that checks for membership in $G_{k,b,x}$ given $k,b,x$ and the trapdoor $t_k$. 
\item There is an efficiently computable injection $\mathcal{J}:\X\to \{0,1\}^w$, such that $\mathcal{J}$ can be inverted efficiently on its range, and such that the following holds. If
\begin{align*}\label{eq:defsetsH}
H_k = \big\{& (b,x_b,d,d\cdot(\mathcal{J}(x_0)\oplus \mathcal{J}(x_1)))\,|\; b\in \{0,1\}, (x_0,x_1)\in \mathcal{R}_k, \\ 
                     &  d\in G_{k,0,x_0}\cap G_{k,1,x_1}\big\},
                    \\
                \overline{H}_k = \{& (b,x_b,d,c)\,|\; (b,x,d,c\oplus 1) \in H_k\big\},
\end{align*}
then for any quantum polynomial-time procedure $\mathcal{A}$ there exists a negligible function $\negl(\cdot)$ such that 
\begin{equation*}
\left|\Pr_{(k,t_k)\leftarrow \textrm{GEN}_{\mathcal{F}}(1^{\lambda})}[\mathcal{A}(k) \in H_k] - \Pr_{(k,t_k)\leftarrow \textrm{GEN}_{\mathcal{F}}(1^{\lambda})}[\mathcal{A}(k) \in\overline{H}_k]\right| \leq \negl(\lambda).
\end{equation*}
\end{enumerate}

	 \end{enumerate}
\end{definition}

\begin{lemma}[{\cite[Theorem 4.1]{brakerski2021cryptographic}}]
Assuming the quantum hardness of LWE, there exists an LWE-based NTCF family with an adaptive hardcore bit property.
\end{lemma}

Here, we recall an inverting algorithm with trapdoor in~\cite{C:MicPei13}, which is used in the LWE-based PoQ scheme~\cite{brakerski2021cryptographic}.
\begin{theorem}[{\cite[Theorem 5.1]{C:MicPei13}}]
\label{thm:trapdoor_samp}
There is an efficient algorithm~$\mathsf{GenTrap}$ that, on input~$1^n, q, m= \Omega(n \log q)$, outputs a matrix~$\*A$ distributed statistically close to uniformly on~$\mathbb{Z}_q^{n \times m}$, and a~$O(m)$-good lattice trapdoor~$\td$ for~$\*A$. Moreover, there is an efficient algorithm
$\mathsf{Invert}$ that, on input~$\*A, \td$ and~$\*s\*A + \*e$ where~$\lVert \*e \rVert \leq q/(C_T\sqrt{n \log q})$
 and~$C_T$ is a universal constant, returns~$\*s$
and~$\*e$ with overwhelming probability over~$(\*A, \td) \gets \mathsf{GenTrap}(1^n,1^m, q)$.
\end{theorem}

\subsection{Extractable-NIZK for NP in the CRS Model}

Notice that, in the following definition from~\cite{JK23}, we additionally consider a weaker notion of extractability, where the adversary is not allowed to query any simulated proof under the targeted CRS. 
We will denote the model simply by extractable-NIZK, where the simulation-extractability is replaced by extractability.
It is easy to see that the simulation-extractable NIZK implies the extractable-NIZK. The latter one is sufficient for our purpose.

\begin{definition}[Post-quantum (simulation-)extractable NIZK for~$\NP$ in CRS Model~\cite{JK23}]
\label{def:simext-nizk-crs}
Let~$\NP$ relation~$\cR$ with corresponding language~$\cL$ be given such that they can be indexed by a security parameter~$\lambda \in \bbN$.

$\Pi = (\Setup, \sP, \sV)$ is a post-quantum (quantum) non-interactive simulation-extractable zero-knowledge argument for~$\NP$ in the CRS model if it has the following syntax and properties.

\noindent \textbf{Syntax:}
The input~$1^\lambda$ is left out when it is clear from context.
\begin{itemize} 
    \item $\crs \gets \Setup(1^\lambda)$: The probabilistic polynomial-size algorithm~$\Setup$ on input~$1^\lambda$ outputs a common reference string~$\crs$.
    \item $\pi \gets \sP(1^\lambda, \crs, x, w)$: The probabilistic (quantum) polynomial-size algorithm~$\sP$ on input a common reference string~$\crs$ and instance and witness pair~$(x, w) \in \cR_\lambda$, outputs a proof~$\pi$.
    \item $\sV(1^\lambda, \crs, x, \pi) \in \zo$: The probabilistic (quantum) polynomial-size algorithm~$\sV$ on input a common reference string~$\crs$, an instance~$x$, and a proof~$\pi$ outputs~$1$ iff~$\pi$ is a valid proof for~$x$.
\end{itemize}

\noindent \textbf{Properties:}
A post-quantum simulation-extractable NIZK is secure if the following properties hold:

\begin{itemize}

    \item {\bf Perfect Completeness.}
    For every~$\lambda \in \bbN$ and every~$(x, w) \in \cR_\lambda$,
    \begin{equation*}
        \Pr_{\substack{\crs \gets \Setup(1^\lambda) \\ \pi \gets \sP(\crs, x, w)}}[\sV(\crs, x, \pi) = 1] = 1.
    \end{equation*}

    \item {\bf Adaptive Multi-Theorem Computational Zero-Knowledge.}
    There exists a probabilistic (quantum) polynomial-size algorithm~$\Sim = (\Sim_0, \Sim_1)$
    
    and a negligible function~$\negl(\cdot)$ such that for every polynomial-size quantum algorithm~$\cA$, and every sufficiently large~$\lambda \in \bbN$,
    {\small
    \begin{align*}
        \mathsf{adv}_{\mathsf{zk}} =&
        \left\vert \Pr_{\substack{\crs \gets \Setup(1^\lambda)}}[\cA^{\sP(\crs, \cdot, \cdot)}(\crs) = 1] - 
        \Pr_{\substack{(\crs, \td) \gets \Sim_0(1^\lambda)}}[\cA^{\Sim_1(\crs, \td, \cdot)}(\crs) = 1] \right\vert 
        \le~\negl(\lambda).
    \end{align*}
    }
    \item {\bf Simulation Extractability.}
    Let~$\Sim = (\Sim_0, \Sim_1)$ be the simulator given by the adaptive multi-theorem computational zero-knowledge property. There exists a polynomial-time extractor~$\Ext$ and a negligible function~$\negl(\cdot)$ such that for every oracle-aided polynomial-size quantum algorithm~$\cA$ and every~$\lambda \in \bbN$,
    \begin{equation*}
        \Pr_{\substack{(\crs, \td) \gets \Sim_0(1^\lambda) \\ (x, \pi) \gets \cA^{\Sim_1(\crs, \td, \cdot)}(\crs) \\ w \gets \Ext(\crs, \td, x, \pi)}}[\sV(\crs, x, \pi) = 1 \wedge x \not\in Q \wedge (x, w) \not\in \cR] \le \negl(\lambda),
    \end{equation*}
    where~$Q$ is the list of queries from~$\cA$ to~$\Sim_1$.

    \item {\bf Extractability.}
    Let~$\Sim$ be the simulator given by the adaptive multi-theorem computational zero-knowledge property. There exists a polynomial-time extractor~$\Ext$ and a negligible function~$\negl(\cdot)$ such that for every polynomial-size quantum algorithm~$\cA$ and every~$\lambda \in \bbN$,
    \begin{equation*}
        \Pr_{\substack{(\crs, \td) \gets \Sim(1^\lambda) \\ (x, \pi) \gets \cA(\crs)\\ w \gets \Ext(\crs, \td, x, \pi)}}[\sV(\crs, x, \pi) = 1 \wedge (x, w) \not\in \cR] \le \negl(\lambda).
    \end{equation*}
    \end{itemize}
\end{definition}

We emphasize that the extractor in the definition of extractability of the extractable-NIZK proof is purely classical, which is sufficient for our purpose. Looking ahead, one can notice that the classically secure extractable-NIZK proof is already enough for our transformation. However, we still recall the post-quantum secure version of extractable NIZK proof aiming for its potential applications in more functionalities such as certifying qubits and key leasing (as discussed in Section~\ref{sec:contribution}).

Next, we recall a result from~\cite{JK23} about an LWE-based construction satisfying the simulation-extractable NIZK definition.

\begin{theorem}[{\cite[Corollary 4.4]{JK23}}]
    \label{cor:simext-nizk-crs}
    Assuming the polynomial quantum hardness of LWE, there exists a post-quantum non-interactive simulation-extractable, adaptively multi-theorem computationally zero-knowledge argument for~$\NP$ in the common reference string model (Definition~\ref{def:simext-nizk-crs}).
\end{theorem}

Note that the construction provided by~\cite{JK23} achieves a stronger notion than the one we need, but it is the closest one that we are aware of. Therefore, we take it as an example of extractable-NIZK for instantiating our transformation.

 \subsection{Recall Approaches for Proofs of Quantumness}

We first recall the definition of the quantum-prover interactive proof.
\begin{definition}[QPIP: Quantum-prover interactive proof]
\label{QPIP}
A quantum-prover interactive proof (QPIP) system is an interactive scheme between two polynomially bounded parties, a quantum prover and a classical verifier interacting over a classical channel.
\begin{enumerate}
    \item A semi-honest QPIP system is described by a pair of algorithms: the PPT algorithm of the honest-but-curious verifier~$\mathcal{V}$ and the QPT algorithm of prover~$\mathcal{P}$;

    \item A malicious QPIP system described by a pair of algorithms: the PPT algorithm of the malicious verifier~$\mathcal{V}^*$ and the QPT algorithm of prover~$\mathcal{P}$;
\end{enumerate}
    
\end{definition}

We first state the straightforward quantum completeness and classical soundness of the proofs of quantumness (PoQ) scheme based on factoring.

\begin{lemma}[Quantum Completeness]
    Due to Shor's quantum polynomial-time factoring algorithm, a QPT prover,~$\mathcal{P}$, following the honest strategy in the scheme~$\Sigma_{\mathsf{Factoring}}$ shown in Figure~\ref{fig:PoQ_factoring} is accepted with probability~$1-\negl(\lambda)$.
        
    \end{lemma}

    \begin{lemma} [Classical Soundness]
        Assuming large integer factoring is classically intractable for any PPT prover~$\mathcal{P'}$, the prover~$\mathcal{P'}$ can convince the~$\mathcal{V}$ to accept with only negligible probability in scheme~$\Sigma_{\mathsf{Factoring}}$ shown in Figure~\ref{fig:PoQ_factoring}.
    \end{lemma}

\begin{figure}[htbp]
\begin{framed}
\centering
\begin{minipage}{1.0\textwidth}
\begin{center}
    \underline{$\Sigma_{\mathsf{Factoring}}$: PoQ for Factoring~$N'$}
\end{center}

\vspace{1mm}
Fix a security parameter~$\lambda$. Let~\( \mathcal{P} \) denote a quantum prover and~\( \mathcal{V} \) denote a classical verifier, and~$\langle \mathcal{P},\mathcal{V}\rangle$ be a QPIP system for a PoQ scheme.
Let~$N'$ be the product of two primes~$p$ and~$q$.

\vspace{1mm}
\noindent {$\mathcal{V}$}: Prepare two random large primes~$(p,q)$ and compute~$N'$. Send~$N'$ to the~$\mathcal{P}$.

\vspace{1mm}
\noindent {$\mathcal{P}$}: Compute~$(p, q) \gets \mathsf{Q.factoring}(N')$ by using Shor's quantum polynomial time algorithm~$\mathsf{Q.factoring}(N')$. Send~$(p, q)$ to the~$\mathcal{V}$. 

\vspace{1mm}
\noindent {$\mathcal{V}$}: Check the validity of~$(p,q)$ and output~$\langle \mathcal{P},\mathcal{V}\rangle=1$ if it passes; else output~$0$ and abort.

\end{minipage}
\end{framed}
\caption{$\Sigma_{\mathsf{Factoring}}$: Factoring-based Proofs of Quantumness Scheme.}
\label{fig:PoQ_factoring}
\end{figure}

\begin{figure}[htbp]
\begin{framed}
\centering
\begin{minipage}{1.0\textwidth}
\begin{center}
    \underline{$\Sigma_{\mathsf{BCMVV}}$: LWE-based PoQ Scheme (Parallel repetition version)}
\end{center}

\vspace{1mm}
Fix a security parameter~$\lambda$ and an NTCF family~$\mathcal{F} \,=\, \big\{f'_{k,b} : \X\rightarrow \mathcal{D}_{\Y} \big\}_{k\in \mathcal{K},b\in\{0,1\}}$  described by algorithms~$(\textsc{Gen}_{\mathcal{F}},\textsc{Samp}_{\mathcal{F}},\textsc{Inv}_{\mathcal{F}},\textsc{Chk}_{\mathcal{F}})$, assuming that LWE is classically intractable. 
Let~\( \mathcal{P} \) denote a quantum prover and~\( \mathcal{V} \) denote a classical verifier, and~$\langle \mathcal{P},\mathcal{V}\rangle$ be a QPIP system for a PoQ scheme.
Repeat the following steps~$\lambda$ times:

\vspace{1mm}
\noindent {$\mathcal{V}$: Prepare~$(k,t_k)\gets \textsc{Gen}_{\mathcal{F}}(1^\lambda)$ and send~$k$ to the prover~$\mathcal{P}$, where~$k=(\mathbf{A}',\mathbf{A}'\mathbf{s}'+\mathbf{e}')$ and~$t_k=\mathsf{td}_{\mathbf{A}'}$ is the trapdoor}.

\vspace{1mm}
\noindent {$\mathcal{P}$: Run~$\textsc{Samp}_{\mathcal{F}}(1^\lambda)$ and measure the image register to yield a string~$\textbf{y}$, send~$\textbf{y}$ to the verifier~$\mathcal{V}$.

\vspace{1mm}
\noindent {$\mathcal{V}$: Sample a uniformly random challenge bit~$c \xleftarrow{\$} \{0,1\}$} and send~$c$ to the prover~$\mathcal{P}$.

\vspace{1mm}
\noindent {$\mathcal{P}$}: Take in the challenge~$c$, do:
\begin{itemize}
    \item Preimage test (if~$c=0$): Perform a standard basis measurement, return a pair~$(b,\mathbf{x}) \in \{0,1\}\times\{0,1\}^n$ as the proof~$\sigma_0=(b,\mathbf{x})$. 
    \item Equation test (if~$c=1$): Perform a Hadamard basis measurement, return a pair~$(u,\mathbf{d})\in \{0,1\}\times\{0,1\}^n$ as the proof~$\sigma_1=(u,\mathbf{d})$.
    \item Send~$\sigma_c$ to the verifier~$\mathcal{V}$.
    
\end{itemize}

\vspace{1mm}
\noindent {$\mathcal{V}$}}: Take in~$(t_k,\textbf{y},c,\sigma_c)$ do:
\begin{itemize}
    \item Compute~$(\mathbf{x}_0,\mathbf{x}_1)\gets \textsc{Inv}_{\mathcal{F}}(1^\lambda,t_k,\textbf{y})$, where~$\mathbf{x}_1= \mathbf{x}_0 -\mathbf{s}' \mod q$. 
    \item Check the validity of~$\sigma_c$ and output~$\langle \mathcal{P},\mathcal{V}\rangle$, which is defined as 
    \begin{align*} 
        \langle \mathcal{P},\mathcal{V}\rangle:= 
\begin{cases} 
1 & \text{if } c = 0 \text{ and } \textsc{Chk}_{\mathcal{F}}(k, \textbf{y}, b, \mathbf{x}) = 1, \\
1 & \text{if } c = 1, d\in G_{k,0,x_0}\cap G_{k,1,x_1} \text{ and } \mathbf{d}^\top\cdot (\mathbf{x}_0 \oplus \mathbf{x}_1) \bmod 2 = u, \\
0 & \text{otherwise}.
\end{cases}
    \end{align*}

\end{itemize}
\noindent At the end of the~$\lambda$ rounds, if the verifier~$\mathcal{V}$ has not aborted it accepts.
\end{minipage}
\end{framed}
\caption{$\Sigma_{\mathsf{BCMVV}}$: LWE-based Proofs of Quantumness Scheme~\cite{brakerski2021cryptographic}.}
\label{fig:BCMVV}
\end{figure}

Now we recall the quantum completeness and classical soundness of the PoQ scheme in~\cite{brakerski2021cryptographic} as follows.

    \begin{lemma}[Quantum Completeness,~\cite{brakerski2021cryptographic}]
    A QPT prover,~$\mathcal{P}$, following the honest strategy in the scheme~$\Sigma_{\mathsf{BCMVV}}$ (as shown in Figure~\ref{fig:BCMVV}) is accepted with probability~$1-\negl(\lambda)$.
        
    \end{lemma}

    \begin{lemma} [Classical Soundness,~\cite{brakerski2021cryptographic}]
        Assume that LWE is classically intractable. For any PPT prover,~$\mathcal{P'}$, in the parallel repetition version of the scheme~$\Sigma_{\mathsf{BCMVV}}$ (as shown in Figure~\ref{fig:BCMVV}), the prover~$\mathcal{P'}$ convinces the~$\mathcal{V}$ to accept with negligible probability.
    \end{lemma}

\section{Zero-Knowledge Proofs of Quantumness Based on Extractable-NIZK}
\label{Sec:ZKPOQ demonst.}

In this section, we initially present the formal definition of zero-knowledge proofs of quantumness. Subsequently, we introduce two constructions that can be provably secure under this definition.

\subsection{Zero-Knowledge Proofs of Quantumness}

In the following definition, we will supplement the zero-knowledge notion to the existing proofs of quantumness definition. Our zero-knowledge property is also formalized by indistinguishability between real transcript and simulated one. In particular, our zero-knowledge protocol requires indistinguishability between the communication initiated by a quantum prover and the one by a classical prover when interacting with the same verifier.

\begin{definition} [Zero-knowledge proofs of quantumness, ZKPoQ] 
\label{ZKPOQ}
Let~\( \mathcal{P} \) denote a quantum prover and~\( \mathcal{V} \) denote a classical verifier. we say that~$\langle \mathcal{P},\mathcal{V}\rangle$ is a QPIP system for zero-knowledge proofs of quantumness if the following properties are satisfied:
\begin{itemize}
        \item \textbf{Quantum Completeness}: Let~\( \lambda \in \mathbb{N} \) be the security parameter. Given the QPT prover~\( \mathcal{P} \), there exists a negligible function in~\( \lambda \) such that:
   \[
   \Pr[ \langle \mathcal{V}(1^\lambda),\mathcal{P}(1^{\lambda}) \rangle =1  ] \geq 1 - \operatorname{negl}(\lambda)
   \]

        \item \textbf{Classical Soundness}: For any PPT prover~\( \mathcal{P}' \), there exists a negligible function in~\( \lambda \) such that:
   \[
   \Pr [\langle \mathcal{V}(1^\lambda),\mathcal{P'}(1^{\lambda}) \rangle =1 ] \leq  \operatorname{negl}(\lambda)
   \]

        \item \textbf{Computational Zero-Knowledge}: For any PPT verifier~$\mathcal{V^*}$ and QPT prover~$\mathcal{P}$, there exists a PPT simulator algorithm~\( \mathsf{Sim}_{\mathcal{V^*}} \) has assess to oracle~$\mathcal{V^*}$, such that:
   \[
   \mathsf{Sim}_{\mathcal{V^*}}(1^\lambda) \approx_c \mathsf{View}[\langle \mathcal{V^*}(1^\lambda) \leftrightarrow \mathcal{P}(1^\lambda) \rangle]
   \]
   where~\( \approx_c \) represents computational indistinguishability.
    \end{itemize}
    
\end{definition}

\subsection{Construction of a Factoring ZKPoQ Scheme}

\begin{figure}[t]
\begin{framed}
\centering
\begin{minipage}{1.0\textwidth}
\begin{center}
\underline{$\mathsf{ZK}.\Sigma_{\mathsf{Factoring}}$: ZKPoQ for Factoring~$N$}
\end{center}

\vspace{1mm}
Fix a security parameter~$\lambda$. Let~\( \mathcal{P} \) denote a quantum prover and~\( \mathcal{V} \) denote a classical verifier. Let~$\langle \mathcal{P},\mathcal{V}\rangle$ be a QPIP system for ZKPoQ scheme described by algorithms~$(\textsc{Setup}, \textsc{Nizk}, \textsc{Prove}, \textsc{Verify})$.
Let~$N$ be the product of two primes~$p$ and~$q$.
Let~$\Pi = (\Setup, \sP, \sV)$ be an extractable non-interactive, adaptively multi-theorem computationally zero-knowledge (extractable-NIZK) scheme for factoring~$N$.

\vspace{1mm}
\noindent {\underline{\textsc{Setup}}$(1^\lambda)$}:
\begin{itemize}
    \item  $(\crs, \td) \gets \Pi.\Setup(1^\lambda)$.

\end{itemize}

\vspace{1mm}
\noindent {$\mathcal{V}$: \underline{\textsc{Nizk}}$(\crs)$}:
\begin{itemize}
    \item Prepare large primes~$(p,q)$ and compute~$N$.
    \item Compute proof~$\pi \gets \Pi.\sP(\crs, N, (p,q))$ for factoring~$N$.

    \item Send~$(N,\pi)$ to quantum prover~$\mathcal{P}$.
\end{itemize}

\vspace{1mm}
\noindent {$\mathcal{P}$: \underline{\textsc{Prove}}$(\crs, N, \pi)$}:
\begin{itemize}
    \item Compute~$\Pi.\sV(\crs, N, \pi)$ and continue if it passes; else output~$\perp$ and abort. 
    \item Compute quantumness proof~$\sigma \gets \mathsf{Q.factoring}(N)$.
    \item Send~$\sigma = (p,q)$ to classical verifier~$\mathcal{V}$.
\end{itemize}

\vspace{1mm}
\noindent {$\mathcal{V}$: \underline{\textsc{Verify}}$(\sigma)$}:
\begin{itemize}
    \item Check the validity of~$\sigma$ and output~$\langle \mathcal{P},\mathcal{V}\rangle=1$ if it passes; else output~$0$ and abort.
\end{itemize}

\end{minipage}
\end{framed}
\caption{$\mathsf{ZK}.\Sigma_{\mathsf{Factoring}}$: Factoring-based Zero-knowledge Proofs of Quantumness Scheme.}
\label{fig:ZKPOQ-Fact}
\end{figure}

We refer to Figure~\ref{fig:ZKPOQ-Fact} for our factoring-based ZKPoQ scheme.

\begin{theorem}\label{thm:factoring_ZKPoQ}
    Assuming that the LWE problem is classically intractable, the protocol~$\mathsf{ZK}.\Sigma_{\mathsf{Factoring}}$ described in Figure~\ref{fig:ZKPOQ-Fact} is a ZKPoQ  (Definition~\ref{ZKPOQ}) scheme satisfying quantum completeness, classical soundness, and computational zero-knowledge.
\end{theorem}

The Theorem~\ref{thm:factoring_ZKPoQ} follows from the following Lemmata~\ref{Property:QComple-Factor},~\ref{Property:Csound-Factor} and~\ref{Property:ZK-Factor}. The completeness of our factoring-based ZKPoQ scheme is stated below.

\begin{lemma} [Quantum Completeness]
\label{Property:QComple-Factor}
    The scheme~$\mathsf{ZK}.\Sigma_{\mathsf{Factoring}}$ satisfies quantum completeness of ZKPoQ in Definition~\ref{ZKPOQ}.
\end{lemma}
\begin{proof}
   Completeness follows from completeness of the Extractable-NIZK, and completeness of PoQ Scheme~$\Sigma_{\mathsf{Factoring}}$ shown in Figure~\ref{fig:PoQ_factoring}. 
 \end{proof}

The classical soundness of our factoring-based ZKPoQ scheme is given as follows. In particular, the classical soundness of our factoring-based ZKPoQ scheme is based on the zero-knowledge property of the extractable-NIZK proof as well as the classical soundness of the original factoring-based PoQ scheme. 

\begin{lemma} [Classical Soundness]
\label{Property:Csound-Factor}
    The scheme~$\mathsf{ZK}.\Sigma_{\mathsf{Factoring}}$ satisfies the classical soundness of ZKPoQ in Definition~\ref{ZKPOQ}.
\end{lemma}
\begin{proof}
     We proceed by contradiction. Assuming that there exists a PPT adversary, denoted as~\( \mathcal{A} \), who can break the classical soundness of scheme~$\mathsf{ZK}.\Sigma_{\mathsf{Factoring}}$ with a non-negligible advantage~\( \varepsilon_0 \).
     We define the following three games. Let~$\mathsf{adv}_i(\mathcal{A})$ denote the advantage of the adversary~$\mathcal{A}$ in the~$\mathsf{Game}_i$ for~$i=\{0,1,2\}$, respectively.
        
\vspace{1mm}

\begin{itemize}
    \item $\underline{\mathsf{Game}_0}$: Let~$\mathsf{Game}_0$ be the same as the real scheme~$\mathsf{ZK}.\Sigma_{\mathsf{Factoring}}$ in Figure~\ref{fig:ZKPOQ-Fact}.

    \item $\underline{\mathsf{Game}_1}$: Let~$\mathsf{Game}_1$ be the same as~$\mathsf{Game}_0$, except the generation of~$(\crs, \td) \gets \textsc{Setup}(1^\lambda)$ is replaced by~$(\crs, \td) \gets \Sim_0(1^\lambda)$, and the proof~$\pi$ generated by the classical verifier~$\mathcal{V}$ is also replaced by the simulated one~$\pi \gets \Sim_1(\crs, \td, N)$ correspondingly,
    where~$\Sim = (\Sim_0, \Sim_1)$ is the simulator corresponding to the adaptive multi-theorem computational zero-knowledge property. 

    \item $\underline{\mathsf{Game}_2}$: Let~$\mathsf{Game}_2$ be the same as~$\mathsf{Game}_1$, except that the~$N$ is replaced by the~$N'$ generated by the challenger for the soundness of the scheme~$\Sigma_{\mathsf{Factoring}}$. We will also forward the response~$(p,q)$ from the prover to the challenger.

\end{itemize}

\noindent$\underline{\mathsf{Game}_0 \approx_c \mathsf{Game}_1}$: This follows directly from the property of the adaptive multi-theorem computational zero-knowledge (Definition~\ref{def:simext-nizk-crs}). Thus,~$\mathsf{adv}_0(\mathcal{A})\leq \mathsf{adv}_1(\mathcal{A})+\mathsf{adv}_{\mathsf{zk}}\leq \mathsf{adv}_1(\mathcal{A})+\negl(\lambda)$.

\vspace{2mm}
\noindent$\underline{\mathsf{Game}_1 \equiv \mathsf{Game}_2}$:  The~$\mathsf{Game}_1$ and~$\mathsf{Game}_2$ are identical as the distributions of~$N$ in the ZKPoQ scheme and~$N'$ from the (honest) challenger for the original PoQ scheme are the same. Thus, we have~$\mathsf{adv}_2(\mathcal{A})=\mathsf{adv}_1(\mathcal{A})$.

First, we claim that~$\mathsf{adv}_2(\mathcal{A}) \leq \varepsilon'$, where the latter one denotes the advantage of any classical PPT adversary~$\mathcal{A}'$ against the soundness of the original PoQ scheme~$\Sigma_{\mathsf{Factoring}}$.
This is because a correct response~$(p,q)$ (as factorization of~$N'$) the adversary~$\mathcal{A}$ in~$\mathsf{Game}_2$ will also be the correct solution to the challenge provided by the challenger for the soundness of the original PoQ scheme. Thus, we have~$\varepsilon_0 = \mathsf{adv}_0(\mathcal{A}) \leq \mathsf{adv}_1(\mathcal{A})+\negl({\lambda})=\mathsf{adv}_2(\mathcal{A})+\negl({\lambda}) \leq \varepsilon'+\negl({\lambda})$. Based on the soundness of the original PoQ scheme, we have~$\varepsilon' \leq \negl(\lambda)$ and therefore~$\varepsilon_0 \le \negl(\lambda)$. The lemma can be concluded by contradiction.
 \end{proof}

Next, we show that our factoring-based ZKPoQ scheme also enjoys the zero-knowledge property. Note that in this security notion, we will consider a malicious verifier who will try to extract unexpected information from the prover. Notably, we can not assume the number~$N$ sent by the malicious verifier will be honestly chosen, which makes the construction of the simulator for zero-knowledge slightly more involved. As mentioned in Section~\ref{sec:contribution}, the zero-knowledge property of our factoring-based ZKPoQ scheme is based on the extractability of the extractable-NIZK proof.

\begin{lemma} [Computational Zero-Knowledge]
\label{Property:ZK-Factor}
    The scheme~$\mathsf{ZK}.\Sigma_{\mathsf{Factoring}}$ satisfies computational zero-knowledge of ZKPoQ in Definition~\ref{ZKPOQ}. 
\end{lemma}

 \begin{proof} In the following, we provide the construction of the PPT simulator~\( \mathsf{Sim}_{\mathcal{V^*}} \) depending solely on the classical verifier~$\mathcal{V^*}$.
 
\vspace{1mm}
 \noindent  \( \mathsf{Sim}_{\mathcal{V^*}} (1^\lambda) \):
 \begin{itemize}
     \item Run~$(\crs, \td) \gets \Sim(1^\lambda)$ and store the~$\td$.

     \item Run verifier~${\mathcal{V^*}}$ on~$\crs$ to get~$N$ and extractable-NIZK proof~$\pi$ for proving the NP relation~$(N,(p,q))$.

     \item Run polynomial-time (classical) extractor~$\Ext$ to compute~$w \gets \Ext(\crs, \td, x, \pi)$, where the witness~$w = (p,q)$ and~$x = N$. If this fails, then we abort.

     \item  Feed the verifier~${\mathcal{V^*}}$ with~$(p, q)$ and the simulation is finished.
 \end{itemize}

 First, the distributions~$\crs$ from the setup algorithm~$\Setup(1^\lambda)$ and the simulation algorithm~$\Sim(1^\lambda)$ are computationally indistinguishable (that is implicit in the zero-knowledge property of the extractable-NIZK proof).
 In addition, since the probability of extractable-NIZK's extractor~$\Ext$ failing to extract~$(p, q)$ is negligible, the probability of simulator~\( \mathsf{Sim}_{\mathcal{V^*}} \) abort is also negligible. 
 Thus, we have~$
   \mathsf{Sim}_{\mathcal{V^*}}(1^\lambda) \approx_c \mathsf{View}[\langle \mathcal{V^*}(1^\lambda) \leftrightarrow \mathcal{P}(1^\lambda) \rangle].$
  \end{proof}

\subsection{Construction of the LWE-based ZKPoQ Scheme}

\begin{figure}[htbp]
\begin{framed}
\centering
\begin{minipage}{1.0\textwidth}
\begin{center}
    \underline{$\mathsf{ZK}.\Sigma_{\mathsf{BCMVV}}$: ZKPoQ for LWE-based Scheme (Parallel repetition version)}
\end{center}

\vspace{1mm}
Fix a security parameter~$\lambda$ and an NTCF family~$\mathcal{F} \,=\, \big\{f'_{k,b} : \X\rightarrow \mathcal{D}_{\Y} \big\}_{k\in \mathcal{K},b\in\{0,1\}}$ described by algorithms~$(\textsc{Gen}_{\mathcal{F}},\textsc{Samp}_{\mathcal{F}},\textsc{Inv}_{\mathcal{F}},\textsc{Chk}_{\mathcal{F}})$, assuming the post-quantum hardness of LWE.
Let~\( \mathcal{P} \) denote a quantum prover and~\( \mathcal{V} \) denote a classical verifier. Let~$\langle \mathcal{P},\mathcal{V}\rangle$ be a QPIP system for ZKPoQ scheme described by algorithms~$(\textsc{Setup}, \textsc{Nizk}, \textsc{Prove}, \textsc{Verify})$.
Let~$\Pi = (\Setup, \sP, \sV)$ be an extractable non-interactive, adaptively multi-theorem computationally zero-knowledge (extractable-NIZK) scheme for factoring~$N$. Repeat the following steps~$\lambda$ times:

\noindent {\underline{\textsc{Setup}}$(1^\lambda)$}:~$(\crs, \td) \gets \Pi.\Setup(1^\lambda)$.

\noindent {$\mathcal{V}$: \underline{\textsc{Nizk}}$(\crs)$}:
\begin{itemize}
    \item Prepare~$(k,t_k)\gets \textsc{Gen}_{\mathcal{F}}(1^\lambda)$, where~$k=(\mathbf{A},\mathbf{A}\mathbf{s}+\mathbf{e})$ and~$t_k=\mathsf{td}_{\mathbf{A}}$.
    \item Recover LWE secret~$\mathbf{s}$ from~$k$ via~$t_k$ using algorithm from Lemma~\ref{thm:trapdoor_samp}.
    \item Compute proof~$\pi \gets \Pi.\sP(\crs, (\mathbf{A},\mathbf{A}\mathbf{s}+\mathbf{e}), \mathbf{s})$.

    \item Send~$(k,\pi)$ to the prover~$\mathcal{P}$.
\end{itemize}

\noindent {$\mathcal{P}$: \underline{\textsc{Prove$_1$}}$(\crs, k, \pi)$}:
\begin{itemize}
    \item Compute~$\Pi.\sV(\crs, k, \pi)$ and continue if it passes; else output~$\perp$ and abort. 
    \item Run~$\textsc{Samp}_{\mathcal{F}}(1^\lambda)$ and measure the image register to yield an string~$\textbf{y}$, send~$\textbf{y}$ to the verifier~$\mathcal{V}$. Note that the~$\textbf{y} = f'_{k,b}(\textbf{x})$ is distributed over random~$b\xleftarrow{\$}\{0,1\}$ and~$\mathbf{x}\xleftarrow{\$}\X$.
    
\end{itemize}

\noindent {$\mathcal{V}$: Sample a uniformly random challenge bit~$c \xleftarrow{\$} \{0,1\}$} and send~$c$ to the prover~$\mathcal{P}$

\noindent {$\mathcal{P}$: \underline{\textsc{Prove$_2$}}$(c,\rho)$}:
\begin{itemize}
    \item Preimage test (if~$c=0$): Perform a standard basis measurement, return a pair~$(b,\mathbf{x}) \in \{0,1\}\times\{0,1\}^n$ as the proof~$\sigma_0=(b,\mathbf{x})$. In this case, the response of the prover is a random one of the two preimages~$(b=0,\textbf{x}_0=\textbf{x})$ and~$(b=1,\textbf{x}_1=\textbf{x}_0-\textbf{s})$ of~$\textbf{y}$ under the function~$f_{k,b}'(\textbf{x}_b)$.
    \item Equation test (if~$c=1$): Perform a Hadamard basis measurement, return a pair~$(u,\mathbf{d})\in \{0,1\}\times\{0,1\}^n$ as the proof~$\sigma_1=(u,\mathbf{d})$. In this case, the response of the prover is~$(u,\mathbf{d})$ such that~$\mathbf{d}$ is random and~$u = \mathbf{d}^\top\cdot (\mathbf{x}_0 \oplus \mathbf{x}_1) \bmod 2$.
    \item Send~$\sigma_c$ to the verifier~$\mathcal{V}$.
    
\end{itemize}

\noindent {$\mathcal{V}$: \underline{\textsc{Verify}}$(t_k,\textbf{y},c,\sigma_c)$}:
\begin{itemize}
    \item Compute~$(\mathbf{x}_0,\mathbf{x}_1)\gets \textsc{Inv}_{\mathcal{F}}(1^\lambda,t_k,\textbf{y})$, where~$\mathbf{x}_1=\mathbf{x}_0 - \mathbf{s} \mod q$. 
    \item Check the validity of~$\sigma_c$ and Output~$\langle \mathcal{P},\mathcal{V}\rangle$, which is defined as 
    \begin{align*} 
        \langle \mathcal{P},\mathcal{V}\rangle:= 
\begin{cases} 
1 & \text{if } c = 0 \text{ and } \textsc{Chk}_{\mathcal{F}}(k, \textbf{y}, b, \mathbf{x}) = 1, \\
1 & \text{if } c = 1, d\in G_{k,0,x_0}\cap G_{k,1,x_1}  \text{ and } \mathbf{d}^\top\cdot (\mathbf{x}_0 \oplus \mathbf{x}_1) \bmod 2 = u, \\
0 & \text{otherwise}.
\end{cases}
    \end{align*}

\end{itemize}
\noindent At the end of the~$\lambda$ rounds, if the verifier~$\mathcal{V}$ has not output~$0$ it accepts.
\end{minipage}
\end{framed}
\caption{$\mathsf{ZK}.\Sigma_{\mathsf{BCMVV}}$: LWE-based Zero-Knowledge Proofs of Quantumness Scheme.}
\label{fig:ZKPoQ-BCMVV}
\end{figure}

We refer to Figure~\ref{fig:ZKPoQ-BCMVV} for our LWE-based ZKPoQ scheme.

\begin{theorem}\label{thm:LWE_ZKPoQ}
    Assuming the post-quantum LWE problem, the protocol~$\mathsf{ZK}.\Sigma_{\mathsf{BCMVV}}$ described in Figure~\ref{fig:ZKPoQ-BCMVV} is a ZKPoQ  (Definition~\ref{ZKPOQ}) scheme satisfying quantum completeness, classical soundness, and computational zero-knowledge.
\end{theorem}

The Theorem~\ref{thm:LWE_ZKPoQ} follows from the Lemmata~\ref{Property:QComple-BCMVV},~\ref{Property:Csound-BCMVV} and~\ref{Property:ZK-BCMVV}. The completeness is given as follows.

\begin{lemma} [Quantum Completeness]
\label{Property:QComple-BCMVV}
    The scheme~$\mathsf{ZK}.\Sigma_{\mathsf{BCMVV}}$ satisfies quantum completeness of ZKPoQ in Definition~\ref{ZKPOQ}.
\end{lemma}
\begin{proof}
   Completeness follows from completeness of the extractable-NIZK, and completeness of PoQ scheme~$\Sigma_{\mathsf{BCMVV}}$ shown in Figure~\ref{fig:BCMVV}. 
 \end{proof}

The classical soundness of our LWE-based ZKPoQ scheme is given as follows. To prove the classical soundness of the LWE-based ZKPoQ scheme, we need to rely on both the zero-knowledge property of the extractable-NIZK proof as well as the classical soundness of the original LWE-based PoQ scheme.

\begin{lemma} [Classical Soundness]
\label{Property:Csound-BCMVV}
    The scheme~$\mathsf{ZK}.\Sigma_{\mathsf{BCMVV}}$ satisfies the classical soundness of ZKPoQ in Definition~\ref{ZKPOQ}.
\end{lemma}
\begin{proof}
     Similarly, we proceed by contradiction. Assuming that there exists a probabilistic polynomial-time (PPT) adversary, denoted as~\( \mathcal{A} \), who can break the classical soundness of scheme~$\mathsf{ZK}.\Sigma_{\mathsf{BCMVV}}$ with a non-negligible advantage~\( \varepsilon_0 \). 
     We define the following three games. Let~$\mathsf{adv}_i(\mathcal{A})$ denote the advantage of the adversary~$\mathcal{A}$ in the~$\mathsf{Game}_i$ for~$i=\{0,1,2\}$, respectively.  
        
\vspace{1mm}

\begin{itemize}
    \item $\underline{\mathsf{Game}_0}$: Let~$\mathsf{Game}_0$ be the same as the real scheme~$\mathsf{ZK}.\Sigma_{\mathsf{BCMVV}}$ in Figure~\ref{fig:ZKPoQ-BCMVV}.

    \item $\underline{\mathsf{Game}_1}$: Let~$\mathsf{Game}_1$ be the same as~$\mathsf{Game}_0$, except that the generation of~$(\crs, \td) \gets \textsc{Setup}(1^\lambda)$ is replaced by~$(\crs, \td) \gets \Sim_0(1^\lambda)$, and the proof~$\pi$ generated by the classical verifier~$\mathcal{V}$ is also replaced by the simulated one~$\pi \gets \Sim_1(\crs, \td, k)$ with~$k=(\mathbf{A},\mathbf{A}\mathbf{s}+\mathbf{e})$ correspondingly,
    where~$\Sim = (\Sim_0, \Sim_1)$ is the simulator corresponding to the adaptive multi-theorem computational zero-knowledge property. 

    \item $\underline{\mathsf{Game}_2}$: Let~$\mathsf{Game}_2$ be the same as~$\mathsf{Game}_1$, except that all transcripts other than the extractable-NIZK proof from the verifier to the prover are replaced by the corresponding transcripts from the challenger for the soundness of the original scheme~$\Sigma_{\mathsf{BCMVV}}$. These transcripts will include~$k$ and~$c$. All transcripts from the prover are also forwarded to the challenger correspondingly. These transcripts will include~$\textbf{y}$ and~$\sigma_c$.

\end{itemize}

\noindent$\underline{\mathsf{Game}_0 \approx_c \mathsf{Game}_1}$: This follows directly from the property of the adaptive multi-theorem computational zero-knowledge (Definition~\ref{def:simext-nizk-crs}). Thus,~$\mathsf{adv}_0(\mathcal{A})\leq \mathsf{adv}_1(\mathcal{A})+\mathsf{adv}_{\mathsf{zk}}\leq \mathsf{adv}_1(\mathcal{A})+\negl(\lambda)$.

\vspace{2mm}
\noindent$\underline{\mathsf{Game}_1 \equiv \mathsf{Game}_2}$: The~$\mathsf{Game}_1$ and~$\mathsf{Game}_2$ are identical because the distributions all transcripts except the extractable-NIZK proof from the verifier to the prover in the ZKPoQ scheme and the ones from the (honest) challenger for the original PoQ scheme are the same. Thus, we have~$\mathsf{adv}_2(\mathcal{A})=\mathsf{adv}_1(\mathcal{A})$. 

First, we claim that~$\mathsf{adv}_2(\mathcal{A}) \leq \varepsilon'$, where the latter one denotes the advantage of any classical PPT adversary~$\mathcal{A}'$ against the soundness of the original PoQ scheme~$\Sigma_{\mathsf{BCMVV}}$. This is because once the adversary~$\mathcal{A}$ succeeds in~$\mathsf{Game}_2$, then we can also pass the challenge proposed by the challenger for the soundness of the original PoQ scheme~$\Sigma_{\mathsf{BCMVV}}$. Therefore, we have~$\varepsilon' \geq \mathsf{adv}_2(\mathcal{A})$. Thus,~$\mathsf{adv}_0(\mathcal{A}) \leq \mathsf{adv}_1(\mathcal{A})+\negl({\lambda})=\mathsf{adv}_2(\mathcal{A})+\negl({\lambda}) \leq \varepsilon'+\negl({\lambda})$. Based on the soundness of the original PoQ scheme, we have~$\varepsilon' \le \negl(\lambda)$ and therefore~$\varepsilon_0 \le \negl(\lambda)$. The lemma can be concluded by contradiction.
 
\end{proof}

Next, we proceed to prove that our LWE-based ZKPoQ scheme also enjoys the zero-knowledge property. We follow a similar manner as before to first extract the witness of the extractable-NIZK proof, which can then be used to properly simulate all responses from the quantum prover (as discussed in Section~\ref{sec:contribution}). Therefore, the zero-knowledge property of our factoring-based ZKPoQ scheme is also based on the extractability of the extractable-NIZK proof. 

\begin{lemma} [Computational Zero-Knowledge]
\label{Property:ZK-BCMVV}
    The scheme~$\mathsf{ZK}.\Sigma_{\mathsf{BCMVV}}$ satisfies computational zero-knowledge of ZKPoQ in Definition~\ref{ZKPOQ}.
\end{lemma}

\begin{proof} In the following, we provide the construction of the PPT simulator~\( \mathsf{Sim}_{\mathcal{V^*}} \) depending solely on the classical verifier~$\mathcal{V^*}$.
 
\vspace{1mm}
 \noindent  \( \mathsf{Sim}_{\mathcal{V^*}} (1^\lambda) \):
 \begin{enumerate}
     \item Run~$(\crs, \td) \gets \Sim(1^\lambda)$ and store the trapdoor~$\td$.

     \item Run verifier~${\mathcal{V^*}}$ on~$\crs$ to get~$k=(\mathbf{A},\mathbf{A}\mathbf{s}+\mathbf{e})$ and an extractable-NIZK proof~$\pi$ for proving NP relation~$(k=(\mathbf{A},\mathbf{A}\mathbf{s}+\mathbf{e}),\mathbf{s})$.

     \item Run polynomial-time (classical) extractor~$\Ext$ to compute~$w \gets \Ext(\crs, \td, x, \pi)$, where the witness~$w=\mathbf{s}$ and the statement~$x=(\mathbf{A},\mathbf{A}\mathbf{s}+\mathbf{e})$. If this fails, then we abort.
     
     \item Pick a random pair~$(b,\mathbf{x})$ where~$b\xleftarrow{\$}\{0,1\}$ and~$\mathbf{x}\xleftarrow{\$}\X$. Compute~$\textbf{y} = f'_{k,b}(\mathbf{x})$ classically and feed the verifier~${\mathcal{V^*}}$ with the next input~$\textbf{y}$. This produces~$\textbf{y}$ with the same distribution as in the real scheme.

     \item  Receive from the verifier~${\mathcal{V^*}}$ the next output challenge~$c$, if~$c=0$, feed the verifier~${\mathcal{V^*}}$ with next input~$(b, \mathbf{x})$; otherwise sample a random~$\mathbf{d}\in \{0,1\}^n$, compute~$u=\mathbf{d}^\top\cdot (\mathbf{x}\oplus(\mathbf{x}-\mathbf{s}))\bmod 2$, then
     feed the verifier~${\mathcal{V^*}}$ with the next input~$(u,\mathbf{d})$ as the equation test result. Note that both the pair~$(b, \mathbf{x}-b\cdot \mathbf{s})$ and the pair~$(u,\mathbf{d})$ have the same distribution as in the real scheme.

     \item Repeat the above Step 2--Step 5 procedures~$\lambda$ times. 
 \end{enumerate}

First, the distributions of~$\crs$ from the setup algorithm~$\Setup(1^\lambda)$ and the simulation algorithm~$(\crs, \td) \gets \Sim(1^\lambda)$ are computationally indistinguishable. 
In addition, since the probability of extractable-NIZK's extractor~$\Ext$ failing to extract~$w=\mathbf{s}$ is negligible, the probability of simulator~\( \mathsf{Sim}_{\mathcal{V^*}} \) abort is also negligible. 

One can also verify that the simulated transcripts from the prover to the verifier have exactly the same distribution as the ones in the real protocol. First, it is clear that the~$\textbf{y}=f_{k,b}'(\textbf{x})$ with both~$b$ and~$\textbf{x}$ randomly chosen, has the same distribution as the one produced in the real protocol. Then, in the case of~$c=0$, the~$(b,\textbf{x})$ is a random one of two preimages~$(0,\textbf{x}_b)$ and~$(1,\textbf{x}_1)$ of~$\textbf{y}$ under the function~$f_{k,b}'(\textbf{x}_b)$, as the real case. The last case is when~$c=1$, the~$(u,\mathbf{d})$ enjoys the relation~$u=\mathbf{d}^\top\cdot (\mathbf{x}\oplus(\mathbf{x}-\mathbf{s}))\bmod 2$ for a random~$\mathbf{d}$, therefore also has the same distribution as the real one.
 Thus, we have~$
   \mathsf{Sim}_{\mathcal{V^*}}(1^\lambda) \approx_c \mathsf{View}[\langle \mathcal{V^*}(1^\lambda) \leftrightarrow \mathcal{P}(1^\lambda) \rangle].$
  \end{proof}

\bibliography{template}

@InProceedings{JK23,
author="Jawale, Ruta
and Khurana, Dakshita",
title="Unclonable Non-interactive Zero-Knowledge",
booktitle="Advances in Cryptology -- ASIACRYPT 2024",
year="2025",
publisher="Springer Nature Singapore",
address="Singapore",
pages="94--128",
doi={https://doi.org/10.1007/978-981-96-0947-5_4}
}

@article{wu2021strong,
  title={Strong quantum computational advantage using a superconducting quantum processor},
  author={Wu, Yulin and Bao, Wan-Su and Cao, Sirui and Chen, Fusheng and Chen, Ming-Cheng and Chen, Xiawei and Chung, Tung-Hsun and Deng, Hui and Du, Yajie and Fan, Daojin and others},
  journal={Physical Review Letters},
  volume={127},
  number={18},
  pages={180501},
  year={2021},
  publisher={APS},
doi={https://doi.org/10.1103/physrevlett.127.180501}
}

@article{arute2019quantum,
  title={Quantum supremacy using a programmable superconducting processor},
  author={Arute, Frank and Arya, Kunal and Babbush, Ryan and Bacon, Dave and Bardin, Joseph C and Barends, Rami and Biswas, Rupak and Boixo, Sergio and Brandao, Fernando GSL and Buell, David A and others},
  journal={Nature},
  volume={574},
  number={7779},
  pages={505--510},
  year={2019},
 doi = {10.1038/s41586-019-1666-5},
  publisher={Nature Publishing Group}
}

@article{madsen2022quantum,
  title={Quantum computational advantage with a programmable photonic processor},
  author={Madsen, Lars S and Laudenbach, Fabian and Askarani, Mohsen Falamarzi and Rortais, Fabien and Vincent, Trevor and Bulmer, Jacob FF and Miatto, Filippo M and Neuhaus, Leonhard and Helt, Lukas G and Collins, Matthew J and others},
  journal={Nature},
  volume={606},
  number={7912},
  pages={75--81},
  year={2022},
  publisher={Nature Publishing Group UK London},
doi={https://doi.org/10.1038/s41586-022-04725-x}
}

@INPROCEEDINGS{brakerski2021cryptographic,
  author={Brakerski, Zvika and Christiano, Paul and Mahadev, Urmila and Vazirani, Umesh and Vidick, Thomas},
  booktitle={59th Annual Symposium on Foundations of Computer Science}, 
  title={A Cryptographic Test of Quantumness and Certifiable Randomness from a Single Quantum Device}, 
  year={2018},
  volume={},
  number={},
  pages={320-331},
   doi={10.1109/FOCS.2018.00038}}

@inproceedings{brakerski2020simpler,
  title={Simpler Proofs of Quantumness},
  author={Brakerski, Zvika and Koppula, Venkata and Vazirani, Umesh and Vidick, Thomas},
  booktitle={15th Conference on the Theory of Quantum Computation, Communication and Cryptography},
  year={2020},
doi={https://doi.org/10.4230/LIPIcs.TQC.2020.8}
}

@inproceedings{aaronson2011computational,
  title={The computational complexity of linear optics},
  author={Aaronson, Scott and Arkhipov, Alex},
  booktitle={43rd Annual ACM Symposium on Theory of Computing},
  pages={333--342},
 doi = {10.1145/1993636.1993682},
  year={2011}}

@article{bremner2011classical,
  title={Classical simulation of commuting quantum computations implies collapse of the polynomial hierarchy},
  author={Bremner, Michael J and Jozsa, Richard and Shepherd, Dan J},
  journal={Proceedings of the Royal Society A: Mathematical, Physical and Engineering Sciences},
  volume={467},
  number={2126},
  pages={459--472},
  year={2011},
  publisher={The Royal Society Publishing},
doi={https://doi.org/10.1098/rspa.2010.0301}}

@inproceedings{shor1994algorithms,
  title={Algorithms for quantum computation: discrete logarithms and factoring},
  author={Shor, Peter W},
  booktitle={35th Annual Symposium on Foundations of Computer Science},
  pages={124--134},
  year={1994},
  organization={IEEE},
doi={10.1109/SFCS.1994.365700}
}

@article{regev2023efficient,
  title={An efficient quantum factoring algorithm},
  author={Regev, Oded},
  journal={Journal of the ACM},
  year={2024},
doi={https://doi.org/10.1145/3708471}
}

@inproceedings{alagic2020non,
  title={Non-interactive classical verification of quantum computation},
  author={Alagic, Gorjan and Childs, Andrew M and Grilo, Alex B and Hung, Shih-Han},
  booktitle={Theory of Cryptography Conference},
  pages={153--180},
  year={2020},
doi = {https://doi.org/10.1007/978-3-030-64381-2_6},
  organization={Springer}
}

@inproceedings{bellare1992defining,
  title={On defining proofs of knowledge},
  author={Bellare, Mihir and Goldreich, Oded},
  booktitle={Advances in Cryptology - {CRYPTO} 1992},
  pages={390--420},
  year={1992},
  organization={Springer},
doi={https://doi.org/10.1007/3-540-48071-4_28}}

@inproceedings{goldreich1986proofs,
  title={Proofs that yield nothing but their validity and a methodology of cryptographic protocol design},
  author={Goldreich, Oded and Micali, Silvio and Wigderson, Avi},
    year={1986},
  booktitle={27th Annual Symposium on Foundations of Computer Science},
doi={doi: 10.1109/SFCS.1986.47}
}

@article{alnawakhtha2024lattice,
  title={Lattice-based quantum advantage from rotated measurements},
  author={Alnawakhtha, Yusuf and Mantri, Atul and Miller, Carl A and Wang, Daochen},
  journal={Quantum},
  volume={8},
  pages={1399},
  year={2024},
doi = {https://doi.org/10.22331/q-2024-07-04-1399},
  publisher={Verein zur F{\"o}rderung des Open Access Publizierens in den Quantenwissenschaften}
}

@inproceedings{KLVY,
  title={Quantum advantage from any non-local game},
  author={Kalai, Yael and Lombardi, Alex and Vaikuntanathan, Vinod and Yang, Lisa},
  booktitle={55th Annual ACM Symposium on Theory of Computing},
  pages={1617--1628},
doi = {10.1145/3564246.3585164},
  year={2023}
}

@InProceedings{C:MicPei13,
author="Micciancio, Daniele
and Peikert, Chris",
title="Hardness of SIS and LWE with Small Parameters",
booktitle="Advances in Cryptology -- CRYPTO 2013",
year="2013",
publisher="Springer Berlin Heidelberg",
address="Berlin, Heidelberg",
pages="21--39",
isbn="978-3-642-40041-4",
doi={https://doi.org/10.1007/978-3-642-40041-4_2}
}

@InProceedings{brakerski2023simple,
author="Brakerski, Zvika
and Gheorghiu, Alexandru
and Kahanamoku-Meyer, Gregory D.
and Porat, Eitan
and Vidick, Thomas",
title="Simple Tests of Quantumness Also Certify Qubits",
booktitle="Advances in Cryptology -- CRYPTO 2023",
year="2023",
publisher="Springer Nature Switzerland",
address="Cham",
pages="162--191",
isbn="978-3-031-38554-4",
doi={https://doi.org/10.1007/978-3-031-38554-4_6}}

@article{KCVY,
  title={Classically verifiable quantum advantage from a computational Bell test},
  author={Kahanamoku-Meyer, Gregory D and Choi, Soonwon and Vazirani, Umesh V and Yao, Norman Y},
  journal={Nature Physics},
  volume={18},
  number={8},
  pages={918--924},
  year={2022},
  publisher={Nature Publishing Group UK London},
doi={https://doi.org/10.1038/s41567-022-01643-7}
}

@misc{chardouvelis2023quantum,
      author = {Orestis Chardouvelis and Vipul Goyal and Aayush Jain and Jiahui Liu},
      title = {Quantum Key Leasing for {PKE} and {FHE} with a Classical Lessor},
      howpublished = {Cryptology {ePrint} Archive, Paper 2023/1640},
      year = {2023},
      url = {https://eprint.iacr.org/2023/1640}
}


\end{document}